\documentclass[11pt]{amsart}
\usepackage[noadjust]{cite}

\usepackage{amssymb}
\usepackage{amsthm}
\usepackage{hyperref}
\usepackage[arrow, matrix, curve]{xy}
\usepackage{color}
\usepackage{enumitem}

\newcommand{\g}{{\mathfrak g}}

\newcommand{\fe}{{\mathfrak e}}

\newcommand{\fg}{{\mathfrak g}}
\newcommand{\fh}{{\mathfrak h}}

\newcommand{\fk}{{\mathfrak k}}

\newcommand{\fm}{{\mathfrak m}}

\newcommand{\fq}{{\mathfrak q}}
\newcommand{\fp}{{\mathfrak p}}

\renewcommand\sp{\mathfrak {sp}}

\renewcommand{\:}{\colon}
\renewcommand{\1}{\mathbf{1}}

\newcommand{\cD}{\mathcal{D}}
\newcommand{\cE}{\mathcal{E}}
\newcommand{\cF}{\mathcal{F}}

\newcommand{\cH}{\mathcal{H}}

\newcommand{\cM}{\mathcal{M}}
\newcommand{\cN}{\mathcal{N}}
\newcommand{\cO}{\mathcal{O}}

\newcommand{\cR}{\mathcal{R}}
\newcommand{\cS}{\mathcal{S}}
\newcommand{\cT}{\mathcal{T}}
\newcommand{\cU}{\mathcal{U}}

\newcommand{\cW}{\mathcal{W}}

\newcommand\bx{{\bf{x}}}

\newcommand\bz{{\bf{z}}}

\newcommand{\bD}{\mathbb{D}}

\newcommand{\mapright}[1]{\smash{\mathop{\arr}\limits^{#1}}}

\newcommand{\ssssarr}{\hbox to 15pt{\rightarrowfill}}
\newcommand{\sssarr}{\hbox to 20pt{\rightarrowfill}}
\newcommand{\ssarr}{\hbox to 30pt{\rightarrowfill}}
\newcommand{\sarr}{\hbox to 40pt{\rightarrowfill}}
\newcommand{\arr}{\hbox to 60pt{\rightarrowfill}}
\newcommand{\larr}{\hbox to 60pt{\leftarrowfill}}
\newcommand{\Arr}{\hbox to 80pt{\rightarrowfill}}

\newcommand{\eset}{\emptyset}

\newcommand{\dd}{{\tt d}}

\newcommand{\subeq}{\subseteq}
\newcommand{\supeq}{\supseteq}

\newcommand{\into}{\hookrightarrow}
\newcommand{\eps}{\varepsilon}

\newcommand{\R}{{\mathbb R}}
\newcommand{\C}{{\mathbb C}}

\newcommand{\T}{{\mathbb T}}

\newcommand{\bH}{{\mathbb H}}
\newcommand{\bS}{{\mathbb S}}

\renewcommand{\hat}{\widehat}

\renewcommand{\tilde}{\widetilde}

\newcommand{\Aff}{\mathop{{\rm Aff}}\nolimits}

\newcommand{\SL}{\mathop{{\rm SL}}\nolimits}

\newcommand{\AU}{\mathop{{\rm AU}}\nolimits}
\newcommand{\PSL}{\mathop{{\rm PSL}}\nolimits}
\newcommand{\SO}{\mathop{{\rm SO}}\nolimits}
\newcommand{\SU}{\mathop{{\rm SU}}\nolimits}

\newcommand{\U}{\mathop{\rm U{}}\nolimits}

\newcommand{\fsl} {\mathop{{\mathfrak{sl} }}\nolimits}

\newcommand{\su}  {\mathop{{\mathfrak{su} }}\nolimits}
\newcommand{\so}  {\mathop{{\mathfrak{so} }}\nolimits}

\newcommand{\Exp}{\mathop{{\rm Exp}}\nolimits}
\newcommand{\Fix}{\mathop{{\rm Fix}}\nolimits}

\newcommand{\ad}{\mathop{{\rm ad}}\nolimits}
\newcommand{\Ad}{\mathop{{\rm Ad}}\nolimits}

\renewcommand{\Im}{\mathop{{\rm Im}}\nolimits}

\newcommand{\Aut}{\mathop{{\rm Aut}}\nolimits}

\newcommand{\diag}{\mathop{{\rm diag}}\nolimits}
\newcommand{\End}{\mathop{{\rm End}}\nolimits}
\newcommand{\id}{\mathop{{\rm id}}\nolimits}

\newcommand{\Inn}{\mathop{{\rm Inn}}\nolimits}

\newcommand{\dS}{\mathop{{\rm dS}}\nolimits}

\newcommand{\nin}{\noindent} 
\newcommand{\oline}{\overline}

\newcommand{\la}{\langle}
\newcommand{\ra}{\rangle}

\newcommand{\res}{\vert}

\newcommand{\spann}{{\rm span}}

\theoremstyle{plain}
\newtheorem{theorem}{Theorem}[section]
\newtheorem{proposition}[theorem]{Proposition}

\theoremstyle{definition}

\newtheorem{example}[theorem]{Example}

\newcommand{\pmat}[1]{\begin{pmatrix} #1 \end{pmatrix}}

\renewcommand{\phi}{\varphi}

\usepackage{color}

\newcommand\be{{\bf{e}}}

\newcommand\bw{{\bf{w}}}

\newcommand{\sH}{{\sf H}}

\newcommand{\sV}{{\tt V}}
\newcommand{\sE}{{\tt E}}
\newcommand{\sF}{{\tt F}}

\newcommand{\PSU}{\mathop{{\rm PSU}}\nolimits}

\renewcommand{\phi}{\varphi} 

\newcommand{\AdS}{\mathop{{\rm AdS}}\nolimits}

\newcommand{\rU}{\mathrm{U}}
\newcommand{\rO}{\mathrm{O}}
\newcommand{\D}{\mathbb{D}}
\begin{document}

\title{Algebraic Quantum Field Theory and Causal Symmetric Spaces} 

\author{Karl-Hermann Neeb}
\address{Department Mathematik,
Friedrich-Alexander-Universit\"at,
Erlangen-N\"urnberg,\break 
Cauer\-strasse 11,
91058 Erlangen, Germany}
\email{neeb@math.fau.de}

\author{Gestur \'Olafsson} 
\address{Department of Mathematics, Louisiana State University, Baton Rouge, LA 70803, USA}
\email{olafsson@math.lsu.edu}



\thanks{
The research of K.-H. Neeb was partially supported by DFG-grant NE 413/10-1. The research of G. \'Olafsson was partially
supported by Simons grant 586106.
} 
\maketitle

\begin{abstract}
  In this article we review our recent work on the
  causal structure of symmetric spaces and
  related geometric aspects of Algebraic Quantum Field Theory.
  Motivated by some general results on modular groups related to nets
  of von Neumann algebras,
  we focus on Euler elements of the Lie algebra,
  i.e., elements whose adjoint action defines a
  $3$-grading. We  study the wedge regions they determine
  in corresponding causal symmetric spaces and
  describe some methods to construct nets of von Neumann algebras
  on causal symmetric spaces that satisfy abstract versions of the
  Reeh--Schlieder and the Bisognano-Wichmann condition. 
\end{abstract}

\maketitle
\tableofcontents

\section{Introduction} 
\label{sec:1}
 
Recent interest in causal symmetric spaces in relation with 
representation theory arose from their role as analogs of
spacetime manifolds in the context of
Algebraic Quantum Field Theory (AQFT) 
in the sense of Haag--Kastler, where one considers 
{\it nets} of von Neumann algebras $\cM(\cO)$
on a fixed Hilbert space $\cH$,
associated to open subsets $\cO$ in some space-time manifold~$M$ 
(\cite{Ha96}). The hermitian elements of the algebra $\cM(\cO)$ represent
observables  that can be measured in the ``laboratory'' $\cO$.
In our context $M$ need not be a time-oriented
Lorentzian manifold. We only assume the existence of a
field $(C_m)_{m \in M}$ of pointed generating closed convex cones
$C_m \subeq T_m(M)$ that is invariant under a smooth action of a
connected Lie group~$G$ with Lie algebra~$\g$.

On de Sitter space $\dS^d$ and Anti-de Sitter space $\AdS^d$ 
the causal structure is given by a Lorentzian metric,
but in general semisimple 
causal symmetric spaces are pseudo-Riemann but not Lorentzian
(cf.\ Section~\ref{sec:2}).
This allows us to study causality aspects of AQFT
in a highly symmetric context without the need of an invariant Lorentzian form.

One typically requires the following properties:
\begin{itemize}
\item[\rm(I)] Isotony: 
  $\cO_1 \subeq \cO_2$ implies $\cM(\cO_1) \subeq \cM(\cO_2)$. 
\item[\rm(L)] Locality: 
  $\cO_1 \subeq \cO_2^\prime$ implies $\cM(\cO_1) \subeq \cM(\cO_2)^\prime$,
  where $\cO^\prime$ is the {\it causal complement}
  of $\cO$, i.e., the maximal
  open subset that cannot be connected to $\cO$ by causal curves,
  and for a von Neumann
  algebra $\cN \subset B(\cH)$ its commutant is denoted $\cN^\prime$.
  \item[\rm(RS)] Reeh--Schlieder property: There exists a unit vector
  $\Omega\in \cH$ that is {\it cyclic} for every $\cM(\cO)$ with $\cO \not=\eset$,
  i.e., the subspace $\cM(\cO)\Omega$ is dense in $\cH$.
\item[\rm(Cov)] Covariance: 
  There exists a unitary representation
  $U \: G \to \U(\cH)$ such that 
 \[U(g) \cM(\cO) U(g)^{-1} = \cM(g\cO)\quad\text{for } g \in G.\] 
\item[\rm(BW)] Bisognano--Wichmann property: 
  There exists an open subset $W \subeq M$ (called a wedge region) 
  such that   $\Omega$ is also {\it separating} for $\cM(W)$, i.e.,
  the map $\cM(W) \to \cH, A \mapsto A\Omega$ is injective.
  We further assume that there exists an element 
  $h \in \g$,  for which 
  the modular operator  $\Delta$ associated to the pair
  $(\cM(W),\Omega)$ by the
  Tomita--Takesaki Theorem \cite[Thm.~2.5.14]{BR87} (cf.\ \eqref{eq:modop} below)
  satisfies   $\Delta^{-it/2\pi} = U (\exp th)$ for $t\in \R$. 
\item[\rm(Vac)] Invariance of vacuum: $U (g)\Omega = \Omega$   for every $g \in G$. 
\end{itemize}

In this context a natural question is, to which extent such
nets of von Neumann algebras exist on a causal homogeneous space
  $M = G/H$ of a finite-dimensional Lie group~$G$,
  and in particular on causal symmetric spaces.
  To address this question, it is natural to simplify the structures
  by considering instead of the pair $(\cM,\Omega)$ the
  corresponding real subspace $\sV := \sV_{(\cM,\Omega)}
  := \oline{\cM_h\Omega}$, where $\cM_h =\{  M \in \cM \: M^* =M\}$.
  This subspace is called
  \begin{itemize}
  \item {\it cyclic} if $\sV + i \sV$ is dense in $\cH$,
    which means that $\Omega$ is cyclic for $\cM$.
  \item {\it separating} if $\sV \cap i \sV = \{0\}$, 
    which means that $\Omega$ is separating for $\cM$.
  \item {\it standard} if it is cyclic and separating, i.e.,
if $\sV \cap i \sV = \{0\}$ and $\oline{\sV + i \sV} = \cH$.
\end{itemize}

These three properties of real subspaces make sense without
  any reference to operator algebras, but they still reflect
  an important part of the underlying structures that can be studied
  in the much simpler context of real subspaces.
If $\sV \subeq \cH$ is a standard subspace, 
then $S_\sV : \sV + i \sV \to \cH, x + i y \mapsto x- iy$ 
defines a densely defined  closed operator with $\sV = \Fix(S_\sV)$; 
the \textit{Tomita operator} of~$\sV$. Its polar decomposition 
can be written as
\begin{equation}
  \label{eq:modop}
  S_\sV = J_\sV \Delta_\sV^{1/2},
\end{equation}
where $J_\sV$ is a {\it conjugation} (an anti-unitary 
involution) and $\Delta_\sV$ is a positive selfadjoint operator 
satisfying $J_\sV \Delta_\sV J_\sV = \Delta_\sV^{-1}$. 
For standard subspaces of the form $\sV_{(\cM,\Omega)}$ we thus recover 
the modular objects provided by the Tomita--Takesaki Theorem. In particular, 
the interaction between unitary group representations and nets of 
von Neumann algebras can already be studied 
in the context of nets of real subspaces. 
A {\it net of real subspaces} on $M$ is a family 
$\sV(\cO) \subeq \cH$ of closed real subspaces of a complex 
Hilbert space $\cH$, assigned to open subsets $\cO$ of a causal 
manifold $M$. We consider the following properties: 
\begin{itemize}
\item[\rm(I)] Isotony: 
$\cO_1 \subeq \cO_2$ implies $\sV(\cO_1) \subeq \sV(\cO_2)$.
\item[\rm(L)] Locality:  $\cO_1 \subeq \cO_2'$ implies
  $\sV(\cO_1) \subeq \sV(\cO_2)' := \sV(\cO_2)^{\bot_\omega}$, where 
$\omega = \Im\la\cdot,\cdot\ra$ is the canonical symplectic form on $\cH$. 
\item[\rm(RS)] Reeh--Schlieder property: 
If $\cO$ is non-empty, then $\sV(\cO)$ is cyclic. 
\item[\rm(Cov)] Covariance:
There exists a unitary representation
  $U \: G \to \U(\cH)$ such that 
$U(g) \sV(\cO)  = \sV(g\cO)$ for $g \in G$. 
\item[\rm(BW)] Bisognano--Wichmann property:
  There exists an open subset $W \subeq M$ (a ``wedge region'')
  such that $\sV(W)$ is standard and an element $h \in \g$ such that 
  $\Delta_{\sV(W)}^{-it/2\pi} = U(\exp t h)$ holds for all $t \in \R$.
\end{itemize}

Nets of standard subspaces serve as building blocks 
  for nets of von Neumann algebras.
Applying second quantization functors (such as bosonic or fermionic
second quantization; see \cite{Si74})
 to associate to each real subspace $\sV \subeq \cH$   
 a pair $(\cR(\sV),\Omega)$, where $\cR(\sV)\subeq B(\cF(\cH))$
 is a von Neumann algebra 
 on a suitable Fock space $\cF(\cH)$, and if $\sV$ is cyclic/separating, then 
 the vacuum vector $\Omega$ is cyclic/separating for $\cR(\sV)$. 
This method has been developed by Araki and Woods  in the context of 
free bosonic quantum fields (\cite{Ar63, Ar64, AW63}); 
some of the corresponding fermionic results are more recent 
(cf.\ \cite{EO73}, \cite{BJL02}). Other statistics (anyons) 
are developed in \cite{Schr97} and more recent deformations
are discussed in \cite[\S 3]{Le15}. 
Thus any net of real subspaces defines a
free quantum field in the sense of Haag--Kastler on the
corresponding Fock space. The properties listed above for 
the net $\cR(\sV(\cO))$ follow 
from the corresponding ones for the net $\sV(\cO)$ and
$\Omega$ is the canonical vacuum vector.

Conversely, any net of von Neumann algebras $\cM(\cO)$ satisfying
  the listed properties immediately leads to the net
  $\sV_{(\cM(\cO), \Omega)}$ of standard subspaces with the corresponding properties.  
  Having these two passages in mind, we therefore content ourselves in the following with the
   discussion of  nets of standard subspaces.

The current interest in standard subspaces arose in the 1990's from the work of Borchers and  Wiesbrock \cite{Bo92, Wi93}.
This led to the concept of modular localization in AQFT 
introduced by Brunetti, Guido and Longo in \cite{BGL02, BGL93};
see also \cite{BDFS00} and \cite{Le15, LL15}
for important applications of this technique.

We know from \cite{MN22b} 
that, if $\ker(U)$ is discrete,
then the elements $h \in \g$ generating the modular group of the 
standard subspace $\sV(W)$ in (BW) is an {\it Euler element}, i.e.,
$\ad h$ defines a $3$-grading
\begin{equation}
  \label{eq:euler1}
 \g = \g_1(h) \oplus \g_0(h) \oplus \g_{-1}(h), 
\quad \mbox{ where } \quad \g_\lambda(h)
= \ker(\ad h - \lambda \1).
\end{equation}
Moreover, the modular conjugation $J = J_{\sV(W)}$ of $\sV(W)$ satisfies
\begin{equation}
  \label{eq:jrel}
  J U(g) J = U(\tau_h(g)) \quad \mbox{ for } \quad
  \tau_h(\exp x) = \exp(e^{\pi i \ad h}x), x \in \g.
\end{equation}
It follows in particular that the unitary representation $U \:  G \to \U(\cH)$
extends by $U(\tau_h) := J$ to an {\it anti-unitary representation} of
the extended group
\begin{equation}
  \label{eq:gext}
  G_{\tau_h} := G \rtimes \{\id_G, \tau_h\}
\end{equation}
 in the sense of \cite{MN21, NO17} (cf.\ Subsection~\ref{sec:4.2}). 

At his point, we are thus facing the following questions:
\begin{itemize}
\item[\rm(Q1)] What are the natural causal spaces $M$ on which such structures
  exist and how does the existence 
of Euler elements relate to the geometry of these spaces? 

\item[\rm(Q2)] Given $M$ and an Euler element $h \in \g$, what are the natural
  wedge regions $W \subeq M$ for which (BW) is satisfied? 
\item[\rm(Q3)] Given an anti-unitary
  representation $(U, \cH)$ of $G_{\tau_h}$, how can we construct
  corresponding nets of standard subspaces?
\item[\rm(Q4)] Which unitary representations $(U,\cH)$ of $G$ occur?
\end{itemize}

For the simplicity of exposition, we restrict in the following to the
special case, where $\g$ is a {\bf simple} real Lie algebra. Many results
are true with suitable modifications in a more general context
(cf.\ \cite{NO21, Oeh22}). 
The existence of an Euler element $h\in \g$ leads to a natural
family of causal symmetric spaces of the form $G/H$, where
$G$ is a connected Lie group with Lie algebra $\g$
(see Section~\ref{sec:2} and \cite{MNO22a} for details).
If, in addition, $\g$ is simple hermitian, i.e., isomorphic to 
\[ \su_{p,q}(\C),   \quad
  \so_{2,d}(\R), \quad
  \sp_{2n}(\R), \quad
  \so^*(2n),\ \quad \fe_{6(-14)} \quad \mbox{ or }  \quad
  \fe_{7(-25)} \]
(see \cite{He78} for the notation),
the Euler element specifies a causal structure on its associated
minimal flag manifold $G/P^-(h)$ (cf.\ Section~\ref{sec:6}) 
whose simply connected covering is a simple space-time manifolds in the sense of 
Mack--de Riese \cite{MdR07}.
We think of this observation as a suitable answer to~(Q1).

On some of these spaces the causal curves define a global order structure
with compact order intervals (they are called {\it globally hyperbolic})
and in this context one can also prove the existence of a global
``time function''  with group theoretic methods (see \cite{Ne91}).
In fact,
if $(G,\tau, H,C)$ is a non-compactly causal symmetric Lie group 
(cf.~Section~\ref{sec:2}), 
then $S=(\exp C)H$ is a closed
subsemigroup, such that the polar map $C \times H \to S$ is a homeo\-morphism
and the causal ordering on $M = G/H$ is given by $aH\ge bH$ if $b^{-1}a\in S$.
In particular the set $\{x\in  G/H\: eH \le x\}$ is homeomorphic to the cone~$C$.
On the other hand, if the center $Z(G)$ is finite, then 
the compactly causal symmetric spaces $G/H$ 
have closed causal curves. Hence there is no global causal ordering
and in particular $\cO'$ cannot be defined in terms of this order;
typically $\cO' = \eset$. 
We refer to the monograph \cite{HO97}
for more details and a complete exposition
of the classification of irreducible symmetric spaces, 
which is also explained in \cite{MNO22a} from the perspective of Euler elements.

On a superficial level, the answer to (Q2) is also rather simple.
Given a causal symmetric space $M = G/H$ with cone field $(C_m)_{m \in M}$ and
the Euler element $h \in \g$, we consider on $M$ the {\it modular flow}
$\alpha_t(gH) = \exp(th)gH$ and write $X_M^h$ for its infinitesimal
generator, a smooth vector field on $M$. Then its {\it positivity domain} 
\begin{equation}
  \label{eq:wmplus}
  W := W_M^+(h) := \{ m \in M \: X_M^h(m) \in C_m^\circ \}
\end{equation}
turns out to be the most natural candidate for a wedge region in~$M$.
However, the geometry of these domains is not so easy to understand, and it
took us some effort to understand these domains
(see Section~\ref{sec:3}, \cite{NO22a, NO22b, MNO22b}).
On the wedge region $W$, the (BW) property has the physical
interpretation that the action of the modular group 
on the algebra $\cM(W)$ is considered to encode the dynamics,
the flow of time, for the quantum theory on $W$  (\cite[\S 6.4]{Su05}).
This connects naturally 
with the approach of Connes and Rovelli who construct the dynamics 
of a quantum statistical system as a modular one-parameter group $\Delta^{it}$ 
(\cite{CR94}). As the flow of time should ``point into the future'',
the positivity domain $W_M^+(h)$ is a natural candidate of a domain for
which (BW) could be satisfied.

Covariant families of real subspaces of $\cH$,  
which are not necessarily standard are easy to construct 
in any anti-unitary representation $(U,\cH)$ of $G_{\tau_h}$ 
using distribution vectors (cf.\ Section~\ref{sec:4}). 
To a finite-dimensional $H$-invariant
real linear subspace  $\sE \subeq \cH^{-\infty}$ and 
an open subset $\cO \subeq M = G/H$,  we associate a 
closed real subspace of $\cH$ as follows. Let $q \: G \to G/H, g \mapsto gH$
be the canonical projection. Then we put 
\begin{equation}
  \label{eq:he1}
  \sH_\sE(\cO)
  := \oline{\spann_\R \{ U^{-\infty}(\phi)\sE \: \phi \in
    C^\infty_c(q^{-1}(\cO),\R)\}} \subeq \cH.
\end{equation}
Here $U^{-\infty}(\phi) \: \cH^{-\infty} \to \cH$ is the linear smoothing 
operator defined by
\[ (U^{-\infty}(\phi)\eta) (\xi )
  = \int_G \phi(g) \eta(U(g^{-1})\xi)\, dg \quad \mbox{ for } \quad
  \xi \in \cH^\infty\]
(cf.\ Section~\ref{sec:4.1}). 
It is easy to see that $\sH_\sE$ defines a net of real subspaces
satisfying (I) and (Cov). So the real problem is to
 find subspaces  $\sE$ such that (RS) and (BW) hold as well
 (cf.\ \cite{NO21}). Typically the elements
$\eta \in \sE$ satisfy a KMS-like condition of the following form:
 The orbit map
 \[ U^\eta \:\R \to \cH^{-\infty}, \quad t \mapsto U^{-\infty}(\exp th)\eta \] 
extends to a map on the closure of the strip
\[ \cS_\pi := \{ z \in\C \: 0 < \Im z < \pi\} \]
which is a holomorphic $\cH$-valued map on the interior,
weak-$*$-continuous on the closed strip and
\begin{equation}
  \label{eq:KMS1}
  U^\eta(t + \pi i) = J U^\eta(t) \quad \mbox{ for } \quad t \in \R.
\end{equation} This implies in particular that
$v := U^\eta\big(\frac{\pi i}{2}\big) \in \cH$ satisfies $J v = v$
and
\begin{equation}
  \label{eq:pi2lim}
  \eta = \lim_{t \to -\frac{\pi}{2}} \Delta^{t/2\pi} v
\end{equation}
in the weak-$*$ topology 
(cf.\ Proposition~\ref{prop:standchar}).

The locality property (L) is even more difficult to implement because
for compactly causal symmetric spaces there may be
closed causal curves, so that
$\cO'$ may be empty and (L) is satisfied for trivial reasons. 
To understand the geometry underlying the locality condition (L)
is an important problem for the future.

We write $\partial U(x)$ for the skew-adjoint infinitesimal generator of the
unitary one-parameter group $U(\exp tx) = e^{t \partial U(x)}$, $ t\in \R$. 
Presently, the best understood situation is the case where the positive cone
\begin{equation}
  \label{eq:cu}
  C_U = \{x\in \fg\: -i \partial U(x) \geq 0 \}
\end{equation}
of the representation $U$ is non-trivial. This relates naturally 
to compactly causal symmetric spaces (cf.\ \cite{NO22a}).
If $\g$ is simple and $U$ is irreducible, \eqref{eq:cu}
specifies unitary highest/lowest weight representations.
We refer to \cite{NO21} for constructions of nets
satisfying (I), (RS), (Cov) and (BW) for the left translation
action of a reductive Lie group $G$ on itself
and also on the flag manifolds $G/P^-(h)$ corresponding to simple space-time
manifolds (cf.\ Section~\ref{sec:6}).
In \cite{Oeh22} one finds generalizations to non-reductive Lie groups.
In \cite{NO22a} one finds a construction for compactly causal symmetric spaces
$M = G/H$ (cf.\ Section~\ref{sec:2}), such as Anti-de Sitter space
$\AdS^d \cong \SO_{2,d-1}(\R)_e/\SO_{1,d-1}(\R)_e$.

Constructions of nets on non-compactly causal spaces are developed in 
\cite{FNO22},
but we are still far from a classification of those $H$-invariant
finite-dimensional subspaces $\sE \subeq \cH^{-\infty}$ for which the net
$\sH_\sE$ satisfies (I), (RS), (Cov) and (BW).
However, we expect that for {\bf every irreducible} unitary representation
$(U,\cH)$ of $G$ such subspaces exist for the corresponding
non-compactly causal symmetric spaces such as 
de Sitter space $\dS^d \cong \SO_{1,d}(\R)_e/\SO_{1,d-1}(\R)_e$.
Here is the central idea of our construction:
Pick a Cartan involution  $\theta$ of $G$ (Section~\ref{sec:2}) with
$\theta(h) = -h$ and consider the subgroup $K := G^\theta$.
Let $\sF \subeq \cH$ be a finite-dimensional $U(K)$-invariant subspace consisting
of $J$-fixed vectors and consider the subspace
\begin{equation}
  \label{eq:betalimit}
 \sE := \beta(\sF) \quad \mbox{ for }\quad
 \beta(v) := \lim_{t \to -\frac{\pi}{2}} e^{it \cdot \partial U(h)}v. 
\end{equation}
Here the main difficulty is to show that $v$ is contained in the domain of the
selfadjoint operators $e^{it \cdot \partial U(h)}$ for $|t| < \pi/2$ and that the
limit $\beta(v)$ exists in the space of distribution vectors of $(U,\cH)$
(cf. \eqref{eq:pi2lim} and \cite{FNO22}). 

In \cite[Thm.~7.1]{MNO22b} we have seen that if $G$ is simple with
trivial center (hence isomorphic to
the adjoint group $\Ad(G) = \Inn(\g)$ of inner automorphisms),
then the positivity domain
$W = W_M^+(h)$ in the corresponding non-compactly causal symmetric space $M = G/H$,
endowed with the maximal invariant cone field,  
is connected and that its stabilizer subgroup
$G_W = \{ g \in G \: g.W = W \}$ is 
$G^h = \{ g \in G \: \Ad(g)h = h\}$. Therefore the {\it wedge space}
$\cW :=  G.W = \{ g.W \:  g \in G\}$ (the set of wedge regions in $M$)
can be identified with the adjoint orbit $\cO_h := \Ad(G)h \subeq \g$
of the Euler element in~$\g$ (Theorem~\ref{thm:6.4}(3)). This connects the above construction
based on distribution vectors with the  abstract
wedge spaces studied in \cite{MN21}. In fact, starting with an anti-unitary
representation $(U,\cH)$ of $G_{\tau_h}$ (see \eqref{eq:gext}), we can associate to the wedge
region $W = W_M^+(h)$ the standard subspace
\begin{equation}
  \label{eq:bgl}
 \sH_{\rm BGL}(W) := \Fix(J\Delta^{1/2}), \quad \mbox{ where } \quad
  J =  U(\tau_h)\quad \mbox{ and } \quad
  \Delta :=  e^{2\pi i \cdot \partial U(h)}.
\end{equation}
From the identity $G_W = G^h$,  it then follows that we obtain a well-defined
covariant extension
\[ \sH_{\rm BGL}(g.W) := U(g) \sH_{\rm BGL}(W).\]
This is called the {\it Brunetti--Guido--Longo (BGL) construction} 
(cf.~\cite{BGL02}, \cite{NO17}). 
Now
\[ \sH_{\rm BGL}(\cO) := \bigcap \{ U(g) \sH_{\rm BGL}(W) \: g\in G,
  \cO \subeq g.W \} \]
defines a net of real subspaces satisfying (I), (Cov) and (BW).
If $\sE \subeq \cH^{-\infty}$ is a real subspace for which the net $\sH_\sE$
in \eqref{eq:he1} satisfies (BW), we immediately obtain the relation
\begin{equation}
  \label{eq:netsrel}
  \sH_{\rm BGL}(\cO) \supeq \sH_\sE(\cO).
\end{equation}
In particular the net $\sH_{\rm BGL}(\cO)$ inherits (RS) from $\sH_\sE(\cO)$.
We presently do not know for which domains $\cO$ and subspaces $\sE$
equality holds in \eqref{eq:netsrel} (see \cite{MN22b} for more details).

The structure of this paper is as follows:
In Section~\ref{sec:2} we discuss the connection between Euler elements
and causal symmetric spaces. We also introduce the crown of the Riemannian
symmetric space $G/K$. In Section~\ref{sec:3} we then turn to the wedge
regions in causal symmetric spaces. 
Section~\ref{sec:4} reviews some concepts related to unitary representations
hat are illustrated by some examples.
Standard subspaces and some of their key properties are presented in
Section~\ref{sec:5}. Finally,  we briefly turn to Cayley type spaces 
and their causal compactifications in Section~\ref{sec:6}.

We hope that this survey provides 
a useful presentation of the fascinating connection
between causal manifold structures, representation theory and
AQFT that is accessible to a wide audience.

\section{Causal symmetric spaces and Euler elements}
\label{sec:2}

\noindent
In this section we discuss causal symmetric spaces and Euler elements
and in particular how Euler elements specify non-compactly causal
symmetric spaces (\cite{MNO22a}). 
First we review basic definitions and concepts related to causal symmetric spaces. On the algebraic level we deal with a
{\it symmetric Lie algebra} $(\g,\tau)$, i.e.,
$\g$ is a finite-dimensional real Lie algebra 
and $\tau$ is an involutive automorphism of~$\g$.
We then have 
\[\g = \fh \oplus \fq \quad \mbox{ with } \quad
  \fh := \fg^{\tau} :=\ker (\tau -\1)\quad\text{and}
  \quad \fq := \fg^{-\tau} := \ker (\tau +\ 1).\]
A {\it causal symmetric Lie algebra}
is a triple $(\g,\tau,C)$, where $(\g,\tau)$ is a symmetric Lie algebra and
$C \subeq \fq$ is a pointed generating 
closed convex cone 
invariant under the group $\Inn_\g(\fh) := \la e^{\ad \fh}\ra$. 
A causal symmetric Lie algebra $(\g,\tau,C)$ is called 
{\it compactly causal} (cc for short) if the 
pointed generating cone $C$ is {\it elliptic}, i.e., 
if its interior consists of elements $x$ which are {\it elliptic} 
in the sense that $\ad x$ is semisimple with purely imaginary spectrum. 
Typical examples of compactly causal Lie algebras 
arise from invariant pointed generating cones $C_\g \subeq \g$
satisfying $-\tau(C_\g) = C_\g$ by $C := \fq \cap C_\fg$.

\nin {\bf Duality of causal symmetric Lie algebras:} 
Causal symmetric Lie algebras come in pairs:  If $(\fg,\tau ,C)$ is causal,
then the {\it $c$-dual symmetric Lie algebra} is
$(\fg^c,\tau^c, iC)$ with 
$\fg^c = \fh \oplus i\fq$ and $\tau^c(x + iy) := x- iy$ is also causal.
This duality exchanges compactly causal and  non-compactly causal  spaces.

On the global level we call a quadruple $(G,\tau,H, C)$
a {\it causal symmetric Lie group} if
$G$ is a connected Lie group, $\tau$ an involutive automorphism of~$G$, 
$H \subeq G^{\tau}$ an open subgroup, and $C \subeq \fq$ a pointed 
generating $\Ad(H)$-invariant closed convex cone.
Then $M := G/H$ is called the associated {\it symmetric space}.
We refer to \cite{He78,Lo69} for the basic theory of symmetric spaces
and to \cite{HO97} for causality aspects.

The involution $\tau$ on $G$
induces an involution $\tau_\fg : \fg \to \fg$ such 
that $\exp (\tau_\fg (x)) = \tau (\exp (x))$ for $x \in \g$  
and $(\g,\tau_\g,C)$ is a causal symmetric Lie algebra.
For simplicity, we shall also write simply $\tau$ instead of $\tau_\g$
on the Lie algebra. 
We say that $M$ is {\it compactly causal}, resp.,
{\it non-compactly causal} if
$(\g,\tau,C)$ has this property. The space
$\fq$ can be identified with the tangent space of $M$ at the origin $x_M=eH$ by
the tangent map 
\[ T_0(q \circ \exp\res_\fq) \: \fq \to T_{x_M}(M), \quad
  x \mapsto \left.\dfrac{d}{dt}\right|_{t=0} \exp (tx).x_M. \] 


If $(G,\tau,H, C) $ is a causal symmetric Lie group, then  $G/H$ becomes a 
{\it causal $G$-space} in the sense that the diffeomorphisms
$\sigma_g(xH) := gxH$ on $G/H$ lead to a $G$-invariant cone field
\[ C_{gH} := T_{x_M}(\sigma_g) C \subeq T_{gH}(M) \quad \mbox{ for } \quad g \in G.\] 
If $G$ is simple, as we are assuming in this article, with finite center,
then, for compactly causal spaces, there are periodic causal curves,
so that no global causal order exists on $M$. However, this pathology
can be resolved by passing to the simply connected covering
$\tilde M$ (\cite{NeuO00}). A typical example is Anti-de Sitter space 
(see Example~\ref{ex:dsads} below). Non-compactly causal spaces always carry a 
global order which is {\it globally hyperbolic} in the sense that
all order intervals are compact (\cite[Thm.5.3.5]{HO97}).

\nin {\bf Cartan involutions:} 
An involutive automorphism $\theta$ of $\g$ is called a  {\it Cartan involution}
if $\Inn_\fg (\fg^\theta)$ is a
maximal compact subgroup of the group $\Inn(\fg)$ of inner automorphisms.
We write 
$\fk = \fg^\theta$, $\fp = \fg^{-\theta}$,
$K := G^\theta$ and note that $G/K$
is a Riemannian symmetric space (see \cite{He78}
for more on the geometry of Riemannian
symmetric spaces).

If $\fm$ is a $\theta$-stable subspace of $\fg$ then we write $\fm_\fk 
= \fm \cap \fk$ and $\fm_\fp= \fm\cap \fp$, so that $\fm = \fm_\fk\oplus \fm_\fp$.
If $\tau : \fg\to \fg$ is an
involution, then there exists a Cartan involution $\theta$
such that $\theta\tau = \tau\theta$ (\cite[Prop.~I.5]{KN96}) and
we thus obtain decompositions
$\fh=\fh_\fk\oplus \fh_\fp$ and $\fq= \fq_\fk \oplus \fq_\fp$.

\nin {\bf Classifying ncc spaces by Euler elements:} 
Let $h \in \g$ be an Euler element. By \cite[Cor.~II.9]{KN96} there exists a
Cartan involution $\theta$ such that $\theta(h) = -h$.
We write $\cE(\fg)\subeq \g$ for the set of Euler elements in~$\g$.
Then $G$ acts on $\cE(\fg)$ by the adjoint action
with finitely many orbits (see \cite[Thm.~3.10]{MN21}
and \cite{Kan98, Kan00} for a classification). 

If $(\g,\tau,C)$ is causal, then an Euler element $h \in C^\circ$
is said to be {\it causal}. 
For any pair $(\theta, h)$ of a Cartan involution $\theta$ and an Euler
element satisfying $\theta(h) = -h$,
the involution $\tau = \tau_h \theta$ with
$\tau_h = e^{\pi i \ad h}$ makes $h$ causal for $(\g,\tau,C)$, where
$C\subeq \fq$ is the closed convex cone generated by $\Inn_\g(\fh)h$.
In \cite[Thm.~4.21]{MNO22a} this construction is used to classify
irreducible non-compactly causal symmetric Lie algebras in
terms of $\Inn(\g)$-orbits of Euler elements. In particular,
for  $(\g, \tau,C)$ 
non-compactly causal, the set of causal Euler elements
is non-empty and contained
in a single $\Inn(\g)$-orbit. 
By duality, this also yields a classification of irreducible cc symmetric
spaces. From the dual perspective, focusing on cc spaces, this classification
goes back to~\cite{Ol91}.

\nin {\bf The crown of the Riemannian symmetric space $G/K$:} 
Throughout our constructions,
analytic continuation of functions and orbit maps play a central role.
In this respect an important tool is the
{\it crown domain of the Riemannian symmetric space $G/K$}.
To describe it, 
we assume that $G$ is contained in a complex Lie group $G_\C$
with a Lie algebra $\fg_\C$.
Let
\begin{equation}
  \label{eq:omega}
  \Omega_\fp = \{x\in \fp \: \sigma(\ad x) \subset (-\pi/2,\pi/2)\}, 
\end{equation}
where $\sigma(\ad x)$ denotes the spectrum of $\ad x$, 
and define the {\it crown of $G/K$} by
\begin{equation}
  \label{eq:crown-1}\Xi  = G(\exp i\Omega_\fp) K_\C\subset G_\C/K_\C, 
  \quad \mbox{ where } \quad K_\C = K \exp(i\fk).
\end{equation}
Then $\Xi$ is open in $G_\C/K_\C$ by \cite{KS04,KS05}
and 
every eigenfunction of the algebra of $G$-invariant differential
operators on $G/K$ extends to a holomorphic function on
$\Xi$ (\cite{KrSc09}).
Of particular interest for us is the following result 
of Gindikin and Kr\"otz on the realization of ncc symmetric spaces
$G/H$ in the boundary of a crown domain \cite[Lem.~ 3.4,Thm~3.5]{GK02}:

\begin{theorem} \label{thm:boundaryorb}
  Let $(\g,\tau,C)$ be a simple ncc symmetric Lie algebra,
  specified by the causal
  Euler element    $h \in \fq_\fp$ via $\tau = \theta \tau_h$
  and $G$ a connected Lie group contained in a complex group~$G_\C$
  with Lie algebra $\g_\C$. Let $K_\C = K \exp(i\fk) \subeq G_\C$ 
and  $s_H := \exp\big(\frac{\pi i}{2}h\big)$. Then 
\[ H := s_H K_\C s_H \cap G \subeq G^\tau \]
is an open subgroup, 
\[ x_H :=s_H K_\C\in \partial \Xi \quad \mbox{ and  } \quad G/H \cong G.x_H. \]
\end{theorem}

Thus, up to covering, every ncc symmetric space can be realized
in the boundary of the crown of $G/K$. So the crown can be used
as a tool to translate 
between different symmetric spaces, in particular the
Riemannian symmetric spaces $G/K$ and all ncc spaces of the form $G/H$.
In particular, boundary values of holomorphic  functions on $\Xi$
lead to distributions on $G/H$, and a generalization of this
idea to vector bundles  can be used to  realize unitary representations of 
$G$ in spaces of distributional sections of vector bundles over~$G/H$
(cf.\ \cite{FNO22}).
We also refer to \cite{GKO03,GKO04} and \cite{NO20} for some results in this
direction.

In this context, an important result by Kr\"otz--Stanton \cite{KS04} is that,
if $G$ is contained in a simply connected complex group $G_\C$, 
for any $K$-finite vector in an irreducible unitary $G$-representation
$(U,\cH)$, there exists a holomorphic extension $\hat{U^v}$ of the orbit map $U^v$ to 
\begin{equation}
  \label{eq:holext}
  \hat{U^v} \: G \exp(i\Omega_\fp)K_\C \to \cH.   
\end{equation}
In particular, $e^{it\cdot \partial U(h)}v   = \hat{U^v}(\exp(ith))$ is
defined for   $|t| < \pi/2$
(cf.\ \eqref{eq:betalimit}).

\begin{example} \label{ex:dsads}
  {\rm(De Sitter and Anti de-Sitter space)} 
An important example of a non-compactly causal irreducible symmetric space
is {\it de Sitter space} 
\begin{equation}
  \label{eq:desitter1}
  \dS^d := \{ (x_0,x_1, \ldots, x_d) \in \R^{1,d} \: x_0^2
  -  x_1^2 - \cdots - x_d^2 = - 1\} 
\end{equation}
with 
$G = \SO_{1,d}(\R)_e$, $H = G_{\be_1} \cong  \SO_{1,d-1}(\R)_e$, and 
$C \subeq T_{\be_1}(\dS^d) \cong  \be_1^\bot$ given by 
\[ C = \{  (x_0, 0, x_2, \ldots, x_{d-1},x_d) \:   x_0 \geq 0, x_0^2 \geq
x_2^2 + \cdots + x_{d}^2\},\]
the closed light cone in $\R^{1,d-1}$.

Likewise {\it Anti-de Sitter space} is a compactly causal irreducible
symmetric space 
\begin{equation}
  \label{eq:adesitter1}
  \AdS^d := \{ (x_0,x_1, \ldots, x_d)
  \in \R^{2,d-1} \: x_0^2 + x_1^2 - x_2^2 - \cdots
  - x_d^2= 1\} 
\end{equation}
with 
$G = \SO_{2,d-1}(\R)_e$, $H = G_{\be_1} \cong \SO_{1,d-1}(\R)_e$, and 
$C \subeq T_{\be_1}(\AdS^d) \cong  \be_1^\bot$ given by 
\[ C = \{  (x_0,0, x_2, \ldots, x_d)\: x_0 \geq 0, x_0^2 \geq x_2^2 + \cdots
+ x_d^2\}.\]
These spaces are dual to each other and can both be realized in the 
complex sphere
\[  \dS^d_\C :=  \{ z \in \C^{d+1} \: z_0^2 - z_1^2 - \cdots - z_d^2 = -1\}
\cong \SO_{d+1}(\C)/\SO_d(\C).\] 
In the case of de Sitter space $G^{\dS}/H^{\dS}$, the corresponding
Riemannian symmetric space $G^{\dS}/K^{\dS}$
is naturally realized as a $d$-dimensional 
hyperbolic space in $\dS^d_\C$: 
\[\bH^d := \{i(x_0,\bx)\in \R^{d+1}\: x_0^2 -\bx^2 =1, x_0>0\}
  = \SO_{1,d}(\R)_e.i\be_0\]
isomorphic to $\SO_{1,d}(\R)_e/\SO_d(\R).$
The crown of $\bH^d$ is the domain 
\[ \Xi = \dS^d_\C \cap (\R^{d+1} + i V_+)
  = \{ z = (z_0,\bz) \in \R^{d+1}+iV_+ \: z_0^2 - \bz^2 =-1\}, \]
where 
$V_+ :=\{(x_0,\bx)\in \R^{d+1} \: x_0^2 - \bx^2>0, x_0>0\}$
is the open upper light cone 
(see \cite[Prop.~3.2]{NO20}).
Note that the Euler element $h \in \so_{1,d}(\R)$ defined by
$h.x = (x_1, x_0,0,\ldots, 0)$ satisfies
\[ \exp(zh).i \be_0 = i \cosh(z) \be_0 + i \sinh(z) \be_1 \in \Xi
  \quad \mbox{ for }  \quad |\Im z| < \frac{\pi}{2} \]
and
$\exp\big(\pm\frac{\pi i}{2}h\big).i \be_0 = \mp \be_1 \in \dS^d$, which exhibits
$\dS^d$ as a $\SO_{1,d}(\R)_e$-orbit in $\partial \Xi$ 
(cf.\ Theorem~\ref{thm:boundaryorb}).

As described above, we expect that the embedding of de Sitter space into
the boundary of $\Xi$, combined with the Kr\"otz--Stanton extension
technique (cf.~\eqref{eq:holext}), 
can be used to realize all irreducible unitary representation of
the Lorentz group in distributional sections of vector bundles
over $\dS^d$. In the particular case of irreducible representations
with  a $K$-fixed vector (spherical representations),
we are simply dealing with functions.
This case has been addressed with methods based on reflection positivity
on the sphere $\bS^d$ in \cite{NO20}. The realization on $\Xi$ leads a
reproducing kernel space of holomorphic functions
whose kernel is given by a hypergeometric function 
\[\Psi_m((z_0,\bz), (w_0,\bw))
  ={}_2F_1\left(\lambda + \frac{d-1}{2}, -\lambda + \frac{d-1}{2}, \frac{n}{2},
    \frac{1-z_0\oline{w_0} + \bz \oline\bw}{2}\right),\]
where $m > 0$ and 
\[ \lambda :=  \sqrt{\Big(\frac{d-1}{2}\Big)^2-m^2}
\in i \R \cup \Big[0,\frac{d-1}{2}\Big).\] 
For $w = \be_1$, we obtain the $H$-invariant function  
\[\Psi_m((z_0,\bz),\be_1) ={}_2F_1\left(\lambda + \frac{d-1}{2}, -
    \lambda + \frac{d-1}{2}, \frac{n}{2}, \frac{1+z_1}{2}\right)\]
which has distributional boundary values on $\dS^d$
with a singularity at $z_1=1$ (cf.~\cite[Thm.~2.2.4]{GKO04}).
We also refer to \cite{BM96} for closely related results by
J.~Bros and H. Moschella.
\end{example}

\begin{example}[$\fsl_2(\R)\cong\so_{1,2}(\R)\cong \su_{1,1}(\C)$]\label{ex:SL2}
(a) The lowest dimensional example of a simple 
causal symmetric Lie algebra arises for the $3$-dimensional
Lie algebra  $\fg= \fsl_2(\R)$ and
  the corresponding connected Lie group $G=\SL_2(\R)$ as follows.
In the basis 
\begin{equation}\label{eq:hef}
h= \frac{1}{2} \begin{pmatrix} 1 & 0 \\ 0 & -1\end{pmatrix},\quad
e = \begin{pmatrix} 0 & 1\\ 0 & 0\end{pmatrix}
\quad\text{and}\quad 
f= \begin{pmatrix} 0 & 0\\ 1 & 0\end{pmatrix}.
\end{equation}
the element $h$ is an Euler element with
\[[h,e] = e,\quad [h,f]=-f\quad\text{and}\quad [e,f] = 2h.\]
Further $\theta(x) = -x^\top$ is a Cartan involution with $\theta(h) = -h$,
so that $\tau := \theta \tau_h$
specifies a non-compactly causal symmetric Lie algebra with 
\[\tau  \: \begin{pmatrix}  a & b\\ c & d\end{pmatrix} \mapsto
  \begin{pmatrix} a & -b\\ -c & d\end{pmatrix}.\]
In particular $\fh = \R h $, $\fq = \R e \oplus \R f$ and
$H:=G^\tau = G^h\cong \R^\times$ is the subgroup of diagonal elements.
Note that the adjoint orbit $\cO_h \cong G/G^h
\cong \SO_{1,2}(\R)_e/\SO_{1,1}(\R)_e$ can be identified
with de Sitter space $\dS^2$ (see \cite{MN22a} for a detailed discussion). 

The cone  generated by $\Inn_\g(\fh)h$ is
$C_{ncc} := [0,\infty) e \oplus [0,\infty)f$ and
the quarter plane $C_{cc} =  [0,\infty) e \oplus [0,\infty)(-f)$ is also
$\Ad(H)$-invariant. So  the symmetric Lie algebra 
$(\fg,\tau ,C_{ncc})$ is non-compactly causal
and $(\fg, \tau, C_{cc})$ is compactly causal.
 
\nin (b) Now replace $\SL_2(\R)$ by its adjoint group
$G:=\SO_{1,2}(\R)_e \cong  \Inn(\g) = \PSL_2(\R)$.
The Cartan involution $\theta$ commutes with $\tau$, so that
$\theta \in G^{\tau}$. As $\theta (e+f) = -(e+f)$,
the cone $C_{ncc}$ is not $G^\tau$ invariant.
On the other hand, $(G,G^\tau,C_{cc})$ is still
compactly causal.
This shows that the causal properties of $(G,H,C)$ depend on the group
$\pi_0(H) = H/H_e$  of connected components of~ $H$.
Note that $G/G^\tau_e = G/G^h\cong \dS^2$ is an example of a causal space of Cayley
type (cf.\ Section~\ref{sec:6}). 

\nin (c) The group $G \cong \PSU_{1,1}(\C) \cong \SO_{1,2}(\R)_e$
acts by M\"obius transformations
\[ \pmat{a & b \\ \oline b & \oline a}.z := \frac{az + b}{-\oline b z + \oline a} \] 
on the unit disc $\bD := \{ z \in \C \: |z| < 1 \}$ and
since $K := \{g \in G \:  g.0 = 0 \}\cong \T$ is maximal compact,
this leads to the identification
$\bD \cong G/K$. The crown domain of $\bD$ can be identified with the
bidisc $\bD^2$ (\cite[Thm.~7.7]{KS05}), into which $\bD$ embeds via $z \mapsto (z, \oline z)$
as a totally real submanifold and $G$ acts on $\bD^2$ by
$g.(z,w) = (g.z, \oline g.w)$.

For the Euler element
$h := \frac{1}{2}\pmat{0 & 1 \\ 1 & 0} \in \su_{1,1}(\C)$, the orbit map
\[ \R \to \bD, \quad t \mapsto \exp(th).0
  = \pmat{ \cosh(t/2) & \sinh(t/2)  \\  \sinh(t/2) & \cosh(t/2)}.0
  = \tanh(t/2) \]
extends to a biholomorphic map
\begin{equation}
  \label{eq:tanh}
 \cS_{\pm \pi/2} = \Big\{ z \in \C \: |\Im z| < \frac{\pi}{2}\Big\}
 \to \bD, \quad z \mapsto \tanh(z/2).
\end{equation}
This implies that, in the closed polydisc $\oline\bD^2$, we have 
\[ \exp\Big(\frac{\pi i}{2}h\Big).(0,0) = i \tan\Big(\frac{\pi}{4}\Big)(1,1)
  = (i,i) \]
and with $H := \{g \in G \: g.i = i, \oline g.i = i \}$ it follows
that $G.(i,i) \cong G/H \cong \dS^2$ is a non-compactly causal symmetric space,
embedded in the Shilov boundary $\T^2$ of $\bD^2$
(cf.\ Theorem~\ref{thm:boundaryorb}). 
\end{example}

  \begin{example} \label{ex:comp-grp} {\rm(Complex case and group case)} 
    (a)  If $\g$ is a complex Lie algebra
    and $\tau$ is antilinear, then $\fg=\fh_\C$ and $\fh$ is a real form of $\g$,
    so that we have $\tau(x+iy) = x-iy$ for  $x,y\in\fh$
    and $\fq = i \fh$.
    We assume that the involution $\tau$ on $\g$
    integrates an involution, also denoted $\tau$,
    on the group $G$. Then, with $H=G^\tau_e$, the
    symmetric space  $G/H$ is causal if and only if
    there exists a pointed generating
    $\Ad(H)$-invariant cone $C \subset \fh$,
    which means that $\fh$ is a simple hermitian Lie algebra.
    Then $C$ is elliptic, so that  $iC $ is a  hyperbolic cone
    in $\fq$, and thus $(\fg,\tau, iC)$ is ncc.

    \nin (b) The $c$-dual $(\g^c,\tau^c)$
    of $(\g,\tau)$ is isomorphic to the symmetric pair $(\fh\oplus \fh, \tau^c_\fg)$ with the flip involution $\tau^c (x,y) = (y,x)$.
            In particular
        \[ (H\times H)^{\tau^c}= \diag (H)=\{(a,a)\: a\in H\}\cong H \]
        is the diagonal subgroup and 
$\fq^c= \{(x,-x)\: x\in \fh\}$. 
The symmetric space $(H\times H)/\diag(H)$ can be identified with the group $H$,
where the $H\times H$-action on $H$ is given by $(a,b).c= acb^{-1}$. 
The elliptic invariant cone $C \subeq \fh$ specifies an elliptic cone
$C^c= \{(x,-x)\: x\in C\}\subeq \fq^c$ and $(H\times H, \diag (H),C^c)$ is compactly causal.
\end{example}

\begin{theorem} {\rm(Cone Extension Theorem)} {\rm(\cite[Thm.~4.5.8]{HO97},
\break \cite[Thm.~2.4]{NO22b})}  \label{thm:ext}
If $(\g,\tau,C)$ is ncc, then
there exists an invariant cone $C_{\fg^c} \subeq \g^c$ in the dual Lie algebra
$\g^c$ such that
\[ C = \fq \cap i C_{\fg^c}. \]
\end{theorem}

\section{Wedge regions in causal symmetric spaces}
\label{sec:3}
\noindent

In this section we discuss several results on wedge regions
associated to Euler elements in a causal symmetric space
and how they can be characterized.

The first domain is  the positivity domain
\[W^+_M(h) :=\{m\in M\: X^M_h(m)\in C^\circ_m\}\]
of  the modular flow (see \eqref{eq:wmplus} in the introduction). 

To specify the second type of domains, we start with the complex {\it tube domain}
\[ \cT_M := G \exp(i C^\circ) \subeq M_\C := G_\C/H_\C \]
(here we assume for simplicity that $G \subeq G_\C$ and $H = H_\C \cap G$)
if $M$ is compactly causal, and 
\[ \cT_M := G \exp(i C_{\rm res}^\circ) \subeq \cT_M \]
if $M$ is non-compactly causal.
Here the domain $C^\circ_{\rm res} \subeq C^\circ$ is specified as follows:
We first note that 
$C^\circ = \Ad(H)(C^\circ \cap \fq_\fp)$. If $h \in C^\circ \cap \fq_\fp$ is a
causal Euler element, then $C^\circ_{\rm res}$ is the unique
$\Ad(H)$-invariant domain specified by 
\begin{align*}
 C^\circ_{\rm res} \cap \fq_\fp
  &= \Big\{ x \in C^\circ \cap \fq_\fp \:
  \sigma\Big(\ad \Big(x - \frac{\pi}{2} h\Big)\Big) \subeq \big(-\frac{\pi}{2},
  \frac{\pi}{2}\big)\Big\}\\
  &  = C^\circ \cap \fq_\fp \cap \Big( \frac{\pi}{2} h + \Omega_\fp\Big).
\end{align*}
For a non-compactly causal space we thus obtain for $G = \Inn(\g)$ a natural
identification of the tube domain $\cT_M$ with the crown domain
$\Xi$ of $G/K$ (see \eqref{eq:crown-1}), both realized in $G_\C/K_\C$ 
(\cite[Thm.~5.4]{NO22b}).

The modular flow
$\alpha_t(m) = \exp(th).m$ extends to a holomorphic flow on $M_\C$, and we
define the {\it KMS wedge domain} by
\[ W_M^{\rm KMS}(h)
  := \{ m \in M \: 
(\forall z \in \cS_\pi)\ \alpha_z(m) \in \cT_M\}.\]
This terminology  is inspired by the close connection with
KMS conditions and standard subspaces;
see \cite{NO19} and Proposition~\ref{prop:standchar} below
which establishes an analogy between the pair $(\sV,\cH)$ for a 
standard subspace $\sV \subeq \cH$ and the pair
$(W_M^{\rm KMS}(h), \cT_M)$.

The following theorem displays some key information on wedge domains.

\begin{theorem} \label{thm:6.4} 
  {\rm(Wedge domains in causal symmetric spaces)} 
Let $(G,\tau^G)$ be a connected symmetric Lie group corresponding 
to the causal symmetric Lie algebra 
$(\g,\tau,C)$. Then the following assertions hold:
\begin{itemize}
\item[\rm(a)] If $(\g,\tau,C)$ is compactly causal
and $h \in \fh$ is an Euler element for which $\tau_h(C) = - C$,
then
$W_M^+(h) = W_M^{\rm KMS}(h)$ and the connected component containing the base point 
$eH$ in its boundary is
\[ G^h_e.\Exp(C_+^\circ + C_-^\circ),\quad \mbox{ where } \quad
C_\pm = \pm C \cap \fq_{\pm 1}(h).\] 
\item[\rm(b)] Suppose that $(\g,\tau,C)$ is non-compactly causal,
  $h \in \fh$ is an Euler element with
  \[ \tau_h(C) = - C, \quad G = \Inn(\g), \quad
    G^c := \Inn_{\g_\C}(\g^c)\quad \mbox{ and } \quad H := G \cap G^c.\]
Then $W_M^{\rm KMS}(h) = G^h_e.\Exp_{eH}(C^c_{\rm res})$
is a connected component of the domain $W_M^+(h)$,
where $C^c_{\rm res} := \{ x_+ + x_- \: x_\pm \in \fq_{\pm 1}(h),
x_+ - x_- \in C^\circ_{\rm res}\}$.
\item[\rm(c)] Suppose that $(\g,\tau,C)$ is non-compactly causal, 
$G = \Inn(\g)$, $C = C^{\rm max}$ is a maximal $\Inn_\g(\fh)$-invariant cone in $\fq$, 
$h \in C^\circ$ a causal Euler element 
  and $H = K^h \exp(\fh_\fp)$. Then the following assertions hold:
  \begin{itemize}
  \item[\rm(1)] The positivity domain $W_M^+(h)$ is connected. 
  \item[\rm(2)] If
    $\gamma(t) := \Exp_{eH}(t h)$ is the causal geodesic defined by $h$, then
$W_M^+(h)$ coincides with the order convex hull $W(\gamma)$ of $\gamma(\R)$.
  \item[\rm(3)] The wedge space
$\cW := \{ gW \: g \in G \}$ 
    of $G$-translates in $M$ is isomorphic to the symmetric space
    $G/G^h \cong \cO_h$.
  \end{itemize}
\end{itemize}
\end{theorem}

\begin{proof} (a) This is a simplification of 
\cite[Thm.~6.5]{NO22a} which is possible
because we assume that $\g$ is simple, so that 
the Extension Theorem~\ref{thm:ext} applies.

\nin (b) \cite[Thms.~6.5,7.1]{NO22b}
  
\nin (c) The connectedness follows from \cite[Thm.~7.1]{MNO22b}, 
the second assertion is \cite[Cor.~7.2]{MNO22b}, and 
the third assertion is \cite[Cor.~7.3]{MNO22b}.
\end{proof}

If $\g$ is a simple hermitian Lie algebra
and $G$ is a corresponding connected Lie group,
then $\g$ contains a pointed generating invariant cone $C_\g$,
and $M := G$ becomes a symmetric space on which the group
$G \times G$ acts by left and right translations.
If $h \in \g$ is  an Euler element, then
$\alpha_t(g) = \exp(th) g\exp(-th)$ is the corresponding modular flow,
and we know from \cite[Thm.~5.2]{NO22a}
that the corresponding wedge domain in $G$ is 
\begin{equation}
  \label{eq:wedgecon}
 W_G^+(h) = W_G^{\rm KMS}(h) = G^h \exp(C_+^\circ + C_-^\circ),
 \ \mbox{ where }  \  
  C_\pm := \pm C_\g \cap \g_{\pm 1}(h).
\end{equation}
In particular, this domain is an open subsemigroup of~$G$. 
For a general irreducible cc space $M = G/H$ and
$C = C_\g \cap \fq$, the Lie algebra $\g$ is hermitian and the map 
\[ Q \: M \to G^{-\tau}_e, \quad \quad gH \mapsto g \tau(g)^{-1}, \] 
defines a covering of $M$ onto a causal subspace 
of the group type space specified by $(\g,C_\g)$. One thus obtains
for $C = C_\g \cap \fq$ the corresponding results in~$M$
(cf.\ Theorem~\ref{thm:6.4}(a)).

For ncc spaces there are stronger results that do not require
an Euler element in $\fh$, which is equivalent to the modular flow on $M$
to have a fixed point,
an assumption that is crucial in Theorem~\ref{thm:6.4}(a,b).
Here is another one that shows all Euler elements with non-trivial
wedge spaces lie in the same orbit.

\begin{theorem} \label{thm:3.2}
  {\rm (\cite[Cor. 6.3, 6.4]{MNO22b})} 
  Let $(G, \tau, H, C)$ be an irreducible
  ncc space, $M = G/H$ 
  and $h \in C^\circ \cap \fq_\fp$ be a causal Euler element with
  $\tau = \theta \tau_h$.  
If $h_1 \in \g$ is an Euler element for which
  the positivity domain $W_M^+(h_1)$ 
  is non-empty, then $h_1 \in \cO_h = \Ad(G)h$.
  In particular, $W_M^+(-h) = \emptyset$
  if $h$ is not symmetric, i.e., $-h \not\in \cO_h$.
\end{theorem}

\section{Unitary and anti-unitary representations}
\label {sec:4}

\subsection{Smooth vectors and distribution vectors}
\label {sec:4.1}

Let $(U,\cH)$ be a unitary representation of the Lie group~$G$. 
A {\it smooth vector} is an element $\eta\in\cH$ for which the orbit map 
$U^\eta : G\to \cH, g \mapsto U(g)\eta$ 
is smooth. We write~$\cH^{\infty}$ for the space 
of smooth vectors. It carries the {\it derived representation} 
$\dd U $ of the Lie algebra $\fg$ given by
\begin{equation}
  \label{eq:derrep}
\dd U(x)\eta =\lim_{t\to 0}\frac{U(\exp t x)\eta -\eta}{t}.
\end{equation}
For write $\partial U(x) = \oline{\dd U(x)}$ for the infinitesimal
generator of the unitary one-parameter group $(U(\exp tx))_{t \in \R}$,
so that $U(\exp tx) = e^{t \partial U(x)}$ for $t \in \R$. 

We extend the $\g$-representation on $\cH^\infty$ to a homomorphism 
$\dd U \:  \cU(\g) \to \End(\cH^\infty),$ 
where $\cU(\g)$ is the complex enveloping algebra of $\g$. This algebra 
carries an involution $D \mapsto D^*$ determined uniquely by 
$x^* = -x$ for $x \in \g$.
For $D \in \cU(\g)$, we obtain a seminorm on $\cH^\infty$ by 
\[p_D(\eta )=\|\dd U(D)\eta\|\quad \mbox{ for } \quad \eta \in \cH^\infty.\] 
These seminorms define a Fr\'echet topology on $\cH^\infty$
for which the inclusion $\cH^\infty\hookrightarrow \cH$ is continuous
(see \cite[Sec.~ A.1]{NO21} for more details). 

The space $\cH^\infty$ of smooth vectors is $G$-invariant 
and we denote the corresponding representation by~$U^\infty$.
We write $\cH^{-\infty}$ for the space 
of continuous anti-linear functionals on $\cH^\infty$. 
Its elements are called \textit{distribution vectors}. 
The group $G$ and the convolution algebra
$C^\infty_c(G)$ act on $\eta \in \cH^{-\infty}$ by
\begin{itemize}
\item $(U^{-\infty}(g)\eta ) (\xi ):= \eta  (U(g^{-1})\xi)$, $g \in G, 
  \xi \in \cH^\infty$.
\item $U^{-\infty}(\varphi) \eta =\int_G \varphi(g) U^{-\infty}(g)\eta\, dg$
  for $\varphi \in C_c^\infty (G).$ 
\end{itemize} 
We have natural $G$-equivariant linear embeddings 
\begin{equation}
  \label{eq:embindistr}
\cH^\infty \into \cH
\mapright{\xi \mapsto  \la \cdot, \xi \ra} \cH^{-\infty}. 
\end{equation}

\subsection{Anti-unitary representations}
\label{sec:4.2}

Let $\cH$ be a complex Hilbert space.
A surjective isometry $A \: \cH \to \cH$ is said to be anti-unitary
if $A(\lambda v) = \oline\lambda Av$ for $v \in \cH,\lambda \in\C$.
We write $\AU (\cH)$ for the group of unitary and
anti-unitary operators on $\cH$.
If $J$ is  an  anti-unitary  involution,
then $\AU (\cH) = \rU (\cH)\dot\cup J\rU(\cH)$ is a disjoint union.

An {\it anti-unitary representation} of a Lie group $G$
is a group homomorphism $U : G\to \AU (\cH)$ for which all
orbit maps $U^v(g) := U(g)v$ are continuous. We refer
to \cite{NO17} for a detailed discussion of anti-unitary extensions
of unitary representations. 
If  $\tau : G\to G$ is an involutive automorphism of the Lie group $G$, then
$G_\tau = G\rtimes \{\id_G,\tau\}$ is a Lie group
and we are interested in extensions of unitary representations of $G$
to $G_\tau$, so that $J := U(\tau)$ is anti-unitary, hence a conjugation 
(an anti-unitary involution).
Any such $J$ commutes with $U(G^\tau)$ and
$\cH^J =\{v\in\cH \: Jv = v\}$ is a closed real subspace invariant under $G^\tau$
such that $\cH = \cH^J \oplus i\cH^J$. In this sense we have
$U (G^\tau) \subset \rO (\cH^J)$ and obtain a real orthogonal representation
of $G^\tau$ on the real Hilbert space~$\cH^J$.

\begin{example} We recall from Example~\ref{ex:SL2} that the group 
  \[G=\SU_{1,1}(\C)=\left\{\begin{pmatrix} a & b\\ \oline b & \oline a\end{pmatrix}\: |a|^2-|b|^2\right\}\]
  acts transitively on $\D=\{z\in \C\: |z|<1\}$ 
by fractional linear transformations 
\[\begin{pmatrix} a & b\\ \oline b & \oline a\end{pmatrix}.z= \frac{az+b}{\oline b z+ \oline a}\]
and the stabilizer of
$0$ is $K\cong \T$.
For $s > 0$ we consider on $\bD$ the positive definite kernel
\[ K(z,w) = \frac{1}{(1 - z \oline w)^s} \]
and the corresponding reproducing kernel Hilbert space
$\cH_s \subeq \cO(\bD)$. Note that $\cH_1 = H^2(\bD)$ is the Hardy space
and $\cH_2 = \cO(\cD) \cap L^2(\cD)$ the Bergman space. 
Then the universal covering group $\tilde G \cong \tilde\SL_2(\R)$ 
of $G$ acts unitarily on $\cH_s$ by lifting the projective
representation of $G$, given by 
\[  (U_s(g) f) (z) = (a + \oline b z)^{-s} f(g^{-1}.z) \quad \mbox{ for } \quad 
 g =\begin{pmatrix} a & b\\ \oline b & \oline a\end{pmatrix}\]
to $\tilde G$. 
The involution $\tau (g) = \overline{g}$ on $G$ lifts to $\tilde G$.
It coincides with $\tau_h$ for the Euler element
$h = \frac{1}{2}\pmat{0 & 1 \\ 1 & 0} \in \so_{1,1}(\R) \subeq \su_{1,1}(\C)$.
The corresponding involution 
on $\D$ is the complex conjugation which induces an conjugation
on $\cH_s$ by 
\[ J(f)(z) = \overline{f(\overline{z})}.\]
Then $J U_s(g) J = U_s(\tau_h(g))$, so that we obtain an anti-unitary
extension of $U_s$ to $\tilde G_{\tau_h}$.
\end{example}

We have the following general extension theorem
for groups with Lie algebra $\fsl_2(\R)$
(\cite[Thm.~4.24]{MN21}): 

\begin{theorem}   Let $h \in \fsl_2(\R)$ be an Euler element
  and $G$ a connected Lie group with Lie algebra $\fsl_2(\R)$.
   Then every unitary  
representation $(U,\cH)$ of $G$ extends to an anti-unitary 
representation of $G_{\tau_h}$ 
on the same Hilbert space. This extension is unique up to isomorphism. 
\end{theorem}

\section{Standard subspaces}
\label{sec:5}


There are two natural ways to  construct standard subspaces. One is by
using a cyclic separating vector $\Omega$ of a von Neumann algebra $\cM$,
so that $\sV_{(\cM,\Omega)} = \oline{\cM_h\Omega}$ is standard.
The other is to use
anti-unitary representations and the  {\it Brunetti--Guido--Longo construction} 
mentioned in the introduction (cf.~\cite{BGL02}, \cite{NO17}). It is based on
the following observation. 

 We have already seen that every standard subspace 
$\sV$ determines a pair $(\Delta_\sV, J_\sV)$ of modular objects 
and that $\sV$ can be recovered from this pair by $\sV = \Fix(J_\sV \Delta_\sV^{1/2})$. We thus obtain a representation theoretic 
parametrization of the set of standard subspaces of $\cH$ 
(cf.~\cite{BGL02}, \cite{NO17}):
Each standard subspace $\sV$ specifies an anti-unitary representation $U^{\sV}$ of $\R^\times$ by  
\begin{equation}
  \label{eq:uv-rep}
U^\sV(e^t) := \Delta_\sV^{-it/2\pi} \quad \mbox{ and }  \quad 
U ^\sV(-1) := J_\sV.  
\end{equation}
This defines a bijection between the set of standard subspaces  and
the set of anti-unitary representations of $\R^\times$ on $\cH$.

The following proposition (\cite[Prop.~2.1]{NOO21}) 
characterizes the elements 
of a standard subspace $\sV$ specified by the pair $(\Delta, J)
= (\Delta_\sV, J_\sV)$ in 
terms of analytic continuation of orbit maps of the unitary 
one-parameter group $(\Delta^{it})_{t \in \R}$ and the real space
$\cH^J$.

\begin{proposition}  \label{prop:standchar} Let $\sV \subeq \cH$ be a standard subspace 
  with modular objects~$(\Delta, J)$,
\[ \cS_\pi := \Big\{ z \in \C \: 0 < \Im z < \pi\Big\} \quad \mbox{  and } \quad 
\cS_{\pm \pi/2} := \{ z \in \C \: |\Im z| < \frac{\pi}{2}\}.\]  For 
$\xi \in \cH$, we consider the orbit map $\alpha^\xi \: \R \to \cH, t \mapsto 
\Delta^{-it/2\pi}\xi$. Then the following are equivalent:
\begin{itemize}
\item[\rm(i)] $\xi \in \sV$. 
\item[\rm(ii)] $\xi \in \cD(\Delta^{1/2})$ with $\Delta^{1/2}\xi = J\xi$. 
\item[\rm(iii)] (KMS-like condition) The orbit map $\alpha^\xi \: \R \to \cH$ 
extends to a continuous map  $\oline{\cS_\pi} \to \cH$ which is 
holomorphic on $\cS_\pi$ and satisfies $\alpha^\xi(\pi i) = J\xi$. 
\item[\rm(iv)] There exists an element $\eta \in \cH^J$ 
whose orbit map $\alpha^\eta$ 
extends to a continuous map $\oline{\cS_{\pm\pi/2}} \to \cH$ which is 
holomorphic on the interior and satisfies $\alpha^\eta(-\pi i/2) = \xi$,
i.e., $\eta \in \cH^J \cap \cD(\Delta^{1/4})$. 
\end{itemize}
\end{proposition}

Let $(U,\cH)$ be an anti-unitary representation of 
$G_{\tau_h} = G \rtimes \{\id_G, \tau_h\}$.
Then
\begin{equation}
  \label{eq:jdelts}
  J =  U(\tau_h)\quad \mbox{ and } \quad  \Delta :=  e^{2\pi i \cdot \partial U(h)}
\end{equation}
specify a standard subspace $\sV := \Fix(J\Delta^{1/2})$.
To construct elements in $\sV$, we want to apply (iv) above.
So we start with a $J$-fixed $K$-finite vector~$v$.
As $(\frac{\pi}{2},\frac{\pi}{2})h \subeq \Omega_\fp$,
the Kr\"otz--Stanton Theorem (cf.~\cite{KS04} and \eqref{eq:holext})
shows that the orbit map
$\alpha^v(t) = U(\exp th)v$ extends holomorphically to the open
strip $\cS_{\pm \pi/2}$, but unfortunately the limit
\[ \beta(v)
  = \lim_{t \to -\frac{\pi}{2}} \alpha^v(it) =
  \lim_{t \to -\frac{\pi}{2}} e^{it \cdot \partial U(h)}v \]
does not exist in $\cH$. However, we expect that it always exists
in the space $\cH^{-\infty}$ of distribution vectors;
see \cite{FNO22} for first steps in this direction.

\begin{example}
  We consider the  strip
  $\cS := \cS_\pi:=  \{z\in \C\: 0 < \Im z < \pi\}$ 
  and its Hardy space
  \[ H^2(\cS) := \Big\{ f \in \cO(\cS) \:
\|f\|^2 :=    \sup_{z \in \cS} \int_\R |f(z + t)|^2\, dt < \infty \Big\}, \]
  where $\cO(\cS)$ is the space of holomorphic functions on $\cS$.
  This is a reproducing kernel Hilbert space with kernel
  \begin{equation}
    \label{eq:kernel-strip} K(z,w) =
    \frac{i}{4\pi\sinh\big(\frac{z-\oline w}{2}\big)}
    \quad \mbox{ satisfying } \quad
    K(z,z) = \frac{1}{4\pi\sin(\Im z)}.
  \end{equation}
The group $\R$ acts unitarily on $H^2(\cS)$ by $(U_tF)(z) = F(z + t)$
  (cf.\ \cite{ANS22}).
  Writing $U_t = e^{it P}$ with the selfadjoint operator $P= -i\frac{d}{dz}$
  on $L^2(\R)$ and $\Delta := e^{-2\pi P}$.
%
  We are interested in the standard subspace corresponding to $\Delta$ and
  the conjugation defined by
  $(J F)(z) := \oline{F(i\pi + \oline z)}$
  via $\sV := \Fix(  J  \Delta^{1/2})$. 
The evaluation functions $K_w(z) := K(z,w)$ satisfy
$JK_w = K_{\pi i + \oline w}$, so that, for $y \in \R$ and
$w = \frac{\pi i}{2} + y$ we have
$K_w  \in \cH^2(\cS)^{  J}$. Consider $\eta := K_{\frac{\pi i}{2}}$. 
For the unitary one-parameter group $(U_t)_{t \in \R}$, we have
$U_t K_w = K_{w-t}$, so that the orbit map 
$\alpha^\eta(t) := K_{\frac{\pi i}{2}-t}$ of $\eta$ 
has an analytic extension $\cS_{\pm \pi/2} \to H^2(\cS)$ given by
$\alpha^\eta(z) = K_{\frac{\pi i}{2} -\oline z}$. It follows in particular that
\[ \|\alpha^\eta(it)\|^2 
  = K\Big(\frac{\pi i}{2} + y +it, \frac{\pi i}{2} + y +it\Big)
  = \frac{1}{4\pi \sin\big(\frac{\pi}{2}+ t\big)} \to \infty\]
  for $|t| \to \frac{\pi}{2}$.
  So the condition in Proposition~\ref{prop:standchar}(iv) is not satisfied.
  However, one can show that $\alpha^\eta(-\frac{\pi i}{2})
  \in H^2(\cS)^{-\infty}$ defines a distribution vector of the
  unitary one-parameter group $(U_t)_{t \in \R}$. It corresponds to the
  tempered distribution on $\partial \cS$, defined on $\R$ by 
  \[ D(x)
  = \lim_{\eps \to 0+}  \frac{i}{4\pi\sinh\big(\frac{x+ i \eps}{2}\big)}
  = \frac{x}{4\pi\sinh\big(\frac{x}{2}\big)}
  \cdot   \lim_{\eps \to 0+} \frac{i}{x + i \eps}\]
  and on $\pi i + \R$ by
  \[ D(\pi i + x)
= \frac{i}{4\pi\sinh\big(\frac{x+ i \pi}{2}\big)}
= \frac{x}{4\pi\cosh\big(\frac{x}{2}\big)}.\]
Then the distribution kernel $D(x-y)$  on $\partial \cS$ is positive
definite, taking boundary values yields an isomorphism of $H^2(\cS)$ with
the corresponding Hilbert subspace $\cH_D \subeq \cD'(\partial \cS)$
and the standard subspace $ \sV$ is generated by the space
$C^\infty_c(\R,\R)$ of real valued test functions on
the lower boundary $\R \subeq \partial \cS$.   
\end{example}

\begin{example} \label{ex:posener}
  {\rm(The group case, \cite{NO21})}\label{ex:GrCase} Let $G$
  be a hermitian Lie group with Cartan involution $\theta$ and
  $(U,\cH)$ be a non-trivial 
  unitary representation of $G$ whose positive cone
  \[C_U = \{x\in \fg\: -i\partial U(x)\geq 0\}\]
  is generating. Then $U$ extends to a representation
  of the semigroup $S_U = G \exp(iC_U)$ that is holomorphic on the interior
  of~$S_U$. We assume that $h \in \g$ is an Euler element and that
  $U$ extends to $G_{\tau_h}$ by $U(\tau_h) = J$ (cf.~\cite{NO17}).
  Then the standard subspace
  $\sV := \sV_{(h,\tau_h)}$ associated to $J$ and $\Delta = e^{2\pi i \cdot
    \partial U(h)}$ as in \eqref{eq:bgl} permits a non-trivial endomorphism 
  semigroup
  \[ S_\sV := \{ g \in G \: U(g) \sV \subeq \sV \}
    = G^{h,\tau_h} \exp(C), \]
  where $C = C_+ + C_-$ for $C_\pm := \pm C_U \cap  \g_{\pm 1}(h)$
  (\cite{Ne19, Ne21}). On the other hand, the interior of this semigroup 
  $S_\sV^\circ$ is a union of connected components of
  the wedge domain of the modular flow
  $\alpha_t(g) = \exp(th) g \exp(-th)$.
  In  \cite{NO21} we use this fact to describe real subspaces
  $\sE \subeq \cH^{-\infty}$ for which $\sH_\sE(S_\sV^\circ)$,
  defined by  \eqref{eq:he1} equals~$\sV$.

  To construct elements of $\sV$ in the spirit of
  Proposition~\ref{prop:standchar}, one can also start with a
  $J$-fixed vector $v \in \cH^J$ and the subsemigroup
  \[ S_U^{\oline \tau_h} = G^{\tau_h} \exp(i(C_+ - C_-)) \]  of
  fixed points of the antiholomorphic extension $\oline\tau_h$
  of $\tau_h$ to $S_U$.
  Then
  \[ U((S_U^{\oline \tau_h})^\circ) \cH^J \subeq \cH^J \cap \cD(\Delta^{1/4}) \] 
  follows easily from
  $\alpha_{\frac{\pi i}{2}}((S_U^{\oline \tau_h})^\circ) = S_\sV^\circ$, so that
  $\Delta^{-1/4} U((S_U^{\oline \tau_h})^\circ) \cH^J \subeq \sV$. 
We refer to \cite{NO22a} for a detailed discussion
of the group  case (see Example~\ref{ex:comp-grp}) 
and for the construction of
nets of standard subspaces satisfying (I), (Cov), (RS) and (BW).
\end{example}

\begin{example} A low dimensional example of a solvable group,
    where many of these phenomena can be
  observed is the affine group
  \[ G = \Aff(\R)_e \cong \R \rtimes \R_+^\times \]
  with Euler element $h = (0,1) \in \g = \R \rtimes \R$ (generator of dilations),
    $\tau_h(b,a) = (-b,a)$ (corresponding to point reflection in $0$),
  and we write $x = (1,0)$ for the generator of translations that
satisfies $[h,x] = x$. 
  For a positive energy representation of $G_{\tau_h}$ we then have
  $C_U = [0,\infty) x = C_+$ and for $\sV = \sV_{(h,\tau_h})$ we have
  $S_\sV = \exp(\R h) \exp(C_+)$.  If $H := \{0\} \rtimes \R^\times_+$ is the
  dilation group, then $M := \R \cong G/H$ is a flat symmetric space,
  the modular flow is given by $\alpha_t(p) = e^t p$ and
  $W_M^+(h) = (0,\infty)$ is the corresponding wedge domain.

  A typical example arises for the Hardy space $H^2(\C_+)$ of the upper half
  plane $\C_+ = \{ z \in \C \: \Im z > 0\}$ with the reproducing kernel
  \[ K(z,w) =
\frac{1}{2\pi} \frac{i}{z-\oline w} \] 
    on which $G$ acts by
  \[ (U(b,a)f)(z) := a^{-1/2} f(a^{-1}(z + b)) \quad \mbox{ and }  \quad
    (Jf)(z) := \oline{f(-\oline z)}.\]
  We may identify $H^2(\C_+)$ by its boundary value map in $L^2(\R)$,
  considered as a Hilbert space $\cH_D$ of (tempered) distributions
  with distribution kernel of the form $D(x-y)$, where 
  \[ D(x) = \lim_{\eps \to 0+} \frac{1}{2\pi} \frac{i}{z+ i \eps}. \]
  From \cite[(3.42)]{NOO21} it follows that the standard subspace
  $\sV \subeq \cH_D$ is generated by the real subspace 
  $e^{-\frac{\pi i}{4}} C^\infty_c((0,\infty),\R) \subeq C^\infty_c(\R,\C)$.
  Note that the support of these test functions is restricted to the wedge
  region $W_M^+(h) = (0,\infty)$. 
\end{example}

\section{Cayley type spaces and causal compactifications}
\label{sec:6}

An irreducible causal  symmetric space  $(G,H,C)$ is
  said to be of {\it Cayley type} if there exists an Euler
  element $h \in \g$ such that $\tau =\tau_h$ and $H = G^h$.
  As $H$ always contains the center, we assume in the following that
  $Z(G) = \{e\}$, i.e., $G  \cong \Inn(\g)$. 
Then $h \in \fh = \g^\tau$, so that $h$ is {\bf not} a causal Euler element.

  Suppose first that $(G,\tau,H,C)$ is non-compactly causal
  and that $h \in \fh$ is an Euler element with $\tau_h = \tau$. Then
  $\fh = \g_0(h)$ and $\fq = \fg_1(h) \oplus \fg_{-1}(h)$. 
  Let $C_\pm = (\pm C) \cap \fg_{\pm 1}(h)$. 
  These are $H$-invariant pointed generating cones in 
  $\fg_{\pm 1}(h)$. 
  The cone $C= C_h := C_+ - C_-$ is hyperbolic and $C_e := C_+ + C_-$ is elliptic.
  In fact, one can characterize (irreducible) Cayley type spaces
  by the existence of   $\Inn_\g(\fh)$-invariant pointed generating
  elliptic and hyperbolic cones in $\fq$ (\cite{HO97}).

  The Lie algebras $\fp^\pm (h)  := \fg_0 (h)\oplus \fg_{\pm 1} (h)$
  are maximal parabolic subalgebras. The
  corresponding maximal parabolic subgroups
  are 
\[ P^\pm (h) = G^h \exp(\fg_{\pm 1}(h)).\] 
As $G \subeq \Aut(\g)$, we have $(G^{\tau_h})_e\subset G^h \subset G^{\tau_h}$. Hence
  $G/G^h$ is a symmetric space which has an open dense orbit in the
  compact manifold
  \[ X := G/P^+ (h)\times G/P^-(h).\]
The cones
  $C_\pm \subeq \g_{\pm 1}(h) \cong \g/\fp^\mp(h)$ are
  $\Ad(P^\mp(h))$-invariant, hence define $G$-invariant causal structures on 
the flag manifolds  $G/P^\mp(h)$. In particular, the product
  $X$ is a causal $G$-manifold,
  and the embedding
  \[ G/H \hookrightarrow G/P^+(h)\times G/P^-(h), \quad
    gH \mapsto (gP^+(h), g P^-(h)) \]
  is causal for $C_h$, resp., $C_e$ if $G/P^-(h)$ carries the causal structure
  defined by $C_+$ and  $G/P^+(h)$ carries the causal structure 
  defined by $-C_-$, resp., $C_-$ 
  (see \cite{NeuO00, OO99, NO22b} for details).

The simply connected coverings of the causal spaces
$G/P^\pm(h)$ coincide with the simple space-time manifolds in the sense of 
Mack--de Riese \cite{MdR07}. As $\g_{\mp 1}(h)$ carries the structure
of a euclidean Jordan algebra, the spaces
$G/P^\pm(h)$ can also be considered as conformal compactifications
of simple euclidean Jordan algebras.
Nets of standard subspaces on the causal manifolds
$G/P^\pm(h)$ satisfying (I), (Cov), (RS) and (BW) 
have been constructed from unitary highest weight representations
in \cite{NO21}. The wedge domain in $G/P^-(h)$ is given
by $W := W^+_{G/P^-(h)}(h) = \exp(C_+^\circ)P^-(h)$ which is diffeomorphic to the
open cone $C_+^\circ \subeq \g_1(h)$, considered as an open subset of~$G/P^-(h)$.

From \cite[\S 7]{MNO22b} we infer that 
\[ G_W = \{ g \in G \: g.W = W\} = G^h,\]
so that the wedge space $\cW := \{ g.W \: g \in G\}$ of
$G/P^-(h)$ can also be identified with the adjoint orbit
$\cO_h$ of~$h$ (cf.\ \cite{MN21}).

In physics the example $\g = \so_{2,d}(\R)$ is of particular
importance. In this case any Euler element $h$ 
leads to $\fh := \g^{\tau_h} \cong \R h \oplus \so_{1,d-1}(\R)$,
$\g_1(h) \cong \R^d$ is $d$-dimensional Minkowski space and
$G/P^-(h) \cong (\bS^1 \times \bS^{d-1})/\{\pm \1\}$ is its conformal
compactification. Here $\cW$ is the space of conformal wedge domains
in Minkowski space.
For $d= 4$, we have in particular 
\[  G/P^-(h) \cong (\bS^1 \times \bS^3)/\{\pm\1\}
\cong (\T \times \SU_2(\C))/\{ \pm\1\} \cong \U_2(\C) \]
on which the group $G = \SU_{2,2}(\C)$ with Lie algebra
$\su_{2,2}(\C) \cong \so_{2,4}(\R)$ acts by conformal maps.

\end{document}